\documentclass[conference]{IEEEtran}
\IEEEoverridecommandlockouts

\title{Near-field Liquid Crystal RIS Phase-Shift Design for Secure Wideband Illumination}

\author{
\IEEEauthorblockN{Mohamadreza Delbari, Qikai Zhou, Robin Neuder, Alejandro Jim\'{e}nez-S\'{a}ez, and Vahid Jamali}
\IEEEauthorblockA{
Technical University of Darmstadt, Darmstadt, Germany
\vspace{-3mm}
\thanks{Delbari, Zhou and Jamali’s work was supported in part by the Deutsche Forschungsgemeinschaft (DFG, German Research Foundation) within the Collaborative Research Center MAKI (SFB 1053, Project-ID 210487104) and in part by the LOEWE initiative (Hesse, Germany) within the emergenCITY Centre under Grant LOEWE/1/12/519/03/05.001(0016)/72. Neuder and Jim\'{e}nez-S\'{a}ez's work was supported by the Deutsche Forschungsgemeinschaft (DFG, German Research Foundation) – Project-ID 287022738 – TRR 196 MARIE within project C09. 
}
}
}

\usepackage{amsfonts,amsmath,amsthm,bm}
\usepackage{cite}
\usepackage{amsmath,amssymb,amsfonts}
\usepackage{algorithmic}
\usepackage{graphicx}
\usepackage{textcomp}
\usepackage{xcolor}
\usepackage[T1]{fontenc}
\usepackage{comment}
\usepackage{eucal}
\usepackage{algorithm}
\usepackage{algorithmic}
\usepackage[acronyms,nonumberlist,nopostdot,nomain,nogroupskip]{glossaries}
\usepackage[font=footnotesize]{caption}
\usepackage{subcaption}
\usepackage{booktabs,flushend}

\newtheorem{lem}{Lemma}

\newcommand{\defeq}{\triangleq}
\def\bOmega{\boldsymbol{\Omega}}

\newcommand{\e}{\mathsf{e}}
\newcommand{\jj}{\mathsf{j}}
\newcommand{\Herm}{\mathsf{H}}
\newcommand{\Trans}{\mathsf{T}}
\newcommand{\x}{\mathsf{x}}
\newcommand{\y}{\mathsf{y}}
\newcommand{\z}{\mathsf{z}}
\newcommand{\bA}{\mathbf{A}}
\newcommand{\bx}{\mathbf{x}}

\newcommand{\bs}{\mathbf{s}}
\newcommand{\bS}{\mathbf{S}}

\newcommand{\bH}{\mathbf{H}}
\newcommand{\ba}{\mathbf{a}}

\newcommand{\bI}{\mathbf{I}}

\newcommand{\bh}{\mathbf{h}}

\newcommand{\bq}{\mathbf{q}}
\newcommand{\bp}{\mathbf{p}}

\newcommand{\bSigma}{\boldsymbol{\Sigma}}
\newcommand{\bGamma}{\boldsymbol{\Gamma}}
\newcommand{\bmu}{\boldsymbol{\mu}}

\newcommand{\blambda}{\boldsymbol{\lambda}}

\newcommand{\Ex}{\mathbb{E}}

\newcommand{\diag}{\mathrm{diag}}

\newcommand{\tr}{\mathrm{tr}}
\newcommand{\rank}{\mathrm{rank}}

\newcommand{\SNR}{\mathrm{SNR}}
\newcommand{\RS}{\mathrm{SR}}
\newcommand{\RIS}{\mathrm{RIS}}

\def\bomega{\boldsymbol{\omega}}

\def\bzero{\boldsymbol{0}}
\def\bone{\boldsymbol{1}}
\def\Cset{\mathbb{C}}

\def\Rset{\mathbb{R}}

\def\tmax{\mathrm{max}}

\def\BS{\mathrm{BS}}
\def\RIS{\mathrm{RIS}}

\def\SNR{\mathrm{SNR}}

\def\eff{\mathrm{eff}}

\def\sCN{\mathcal{CN}}

\def\Pset{\mathcal{P}}
\def\bigO{\mathcal{O}}
\usepackage{xcolor}

\newacronym{RIS}{RIS}{reconfigurable intelligent surface}
\newacronym{QoS}{QoS}{quality of service}
\newacronym{LC}{LC}{liquid crystal}
\newacronym{SNR}{SNR}{signal-to-noise ratio}
\newacronym{TDMA}{TDMA}{time-division multiple-access}
\newacronym{BS}{BS}{base station}
\newacronym{MU}{MU}{mobile user}
\newacronym{ME}{ME}{mobile eavesdropper}
\newacronym{NF}{NF}{near-field}
\newacronym{Tx}{Tx}{transmitter}
\newacronym{Rx}{Rx}{receiver}
\newacronym{AWGN}{AWGN}{additive white Gaussian noise}
\newacronym{wrt}{w.r.t.}{with respect to}
\newacronym{RDE}{RDE}{Reaction-Diffusion Equation}
\newacronym{PDE}{PDE}{partial differential equation}
\newacronym{UPA}{UPA}{uniform planar array}
\newacronym{ULA}{ULA}{uniform linear array}
\newacronym{AO}{AO}{alternative optimization}
\newacronym{SOCP}{SOCP}{second-order cone programming}
\newacronym{AoD}{AoD}{angle of departure}
\newacronym{LOS}{LOS}{line of sight}
\newacronym{nLOS}{nLOS}{non-LOS}
\newacronym{MIMO}{MIMO}{multiple-input multiple-output}
\newacronym{RS}{SR}{secure rate}
\newacronym{SDP}{SDP}{semi-definite programming}
\newacronym{6G}{6G}{sixth generation}
\newacronym{CSI}{CSI}{channel state information}

\newacronym{PIN}{PIN}{positive-intrinsic-negative}
\newacronym{RF}{RF}{radio frequency}
\newacronym{MEMS}{MEMS}{micro-electro-mechanical system}
\newacronym{mmWave}{mmWave}{millimeter wave}
\newacronym{OFDM}{OFDM}{orthogonal frequency division multiplexing}

\begin{document}

\maketitle

\begin{abstract}
\Gls{LC} technology provides a low-power and scalable approach to implement a \gls{RIS}. However, the \gls{LC}-based \gls{RIS}'s phase-shift response is inherently frequency-dependent, which can lead to performance degradation if not properly addressed. This issue becomes especially critical in secure communication systems, where such variations may result in considerable information leakage. To avoid the need for full \gls{CSI} acquisition and frequent \gls{RIS} reconfiguration, we design \gls{RIS} for a wideband \gls{OFDM} system to illuminate a desired area containing legitimate users while avoiding leakage to regions where potential eavesdroppers may be located. Our simulation results demonstrate that the proposed algorithm improves the secrecy rate compared to methods that neglect frequency-dependent effects. In the considered setup, the proposed method achieves a secrecy rate of about 2 bits/symbol over an 8 GHz bandwidth when the center frequency is 60 GHz.
\end{abstract}
\glsresetall
\section{Introduction}
\Glspl{RIS} are considered a promising technology for the future generations of wireless communications, as they enable programmable and controllable radio environments \cite{Wu2019,Basar2019,di2019smart}. \Gls{LC} technology has recently emerged as an energy-efficient and scalable approach for realizing a large \gls{RIS}, especially in \gls{mmWave} communication systems \cite{neuder2024architecture,aboagye2022design,delbari2025fast}.  Recent studies have explored various aspects of \gls{LC}-based \glspl{RIS}, ranging from practical implementations to theoretical analysis. For instance, the compact experimental realization of an \gls{LC}-\gls{RIS} was demonstrated in \cite{neuder2024architecture}, while its potential in visible light communication was investigated in \cite{aboagye2022design}. A broader evaluation of the fundamental properties of \gls{LC}-\gls{RIS}, such as energy efficiency and cost-effectiveness, is presented in \cite{jimenez2023reconfigurable}, which also benchmarks these systems against alternative implementation technologies. The temporal dynamics of \glspl{LC} were modeled in \cite{delbari2024fast,delbari2025fast} in a time-division multiple-access scheme. The authors in \cite{delbari2024temperature,gholian2025temperature} analyzed the impact of the temperature on the performance of \gls{LC}-\gls{RIS}. Collectively, these contributions emphasize the increasing interest in improving the practicality and performance of \gls{LC} components within \gls{RIS}. 

\gls{LC}-based \gls{RIS}s are inherently frequency-dependent \cite{neuder2024architecture}, which can lead to considerable performance degradation in wideband systems. This issue is especially critical in the context of physical layer security, where unintended information leakage may occur, particularly under \gls{OFDM}-based transmission. Several recent studies have addressed the challenges of wideband \gls{RIS}-aided communication in this setting \cite{Yu2024,Qian2025,Cui2024,He2021,Su2023,Mo2024,li2021intelligent}. The problem of beam splitting in wideband \gls{RIS}-aided communication systems has motivated extensive research, resulting in advances in \gls{NF} beamforming \cite{Yu2024,Cui2024}, signal processing techniques \cite{Qian2025}, and system-level architectures \cite{He2021,Su2023,Mo2024}. However, these works often assume idealized \gls{RIS} models, where the phase shift is either frequency-independent or can be perfectly controlled across subcarriers. Even studies that consider more practical models \cite{li2021intelligent} typically rely on precise knowledge of the users' location. 

In contrast, this paper addresses the design of phase shifters under the practical constraints imposed by \gls{LC}-\gls{RIS} technology when the exact locations are not necessarily known. In particular, we design the \gls{RIS} for secure illumination by maximizing the received power in the desired area containing the \glspl{MU}, while minimizing it in the region where potential eavesdroppers may be located. We develop a novel optimization framework that accounts for the distinct frequency response of each \gls{RIS} element to construct a secure communication region. This approach bridges the gap between theoretical wideband beamforming and its hardware-limited realization in secure spatial coverage. To the best of the authors' knowledge, the frequency-dependency of \gls{LC}-\gls{RIS}s in a wideband scenario has not been investigated in the literature, yet. Our key contributions are as follows:

\begin{itemize}
    \item First, we develop a model that accounts for the impact of frequency on the phase-shift response of \gls{LC}-\gls{RIS}. We characterize how altering the frequency influences the response of the phase shifts.
    \item Next, we introduce a novel secrecy setup involving an \gls{LC}-\gls{RIS}-assisted system with multiple legitimate users distributed within a given area and a mobile eavesdropper. In this setup, the exact location and instantaneous \gls{CSI} of the eavesdropper (and the users) are assumed to be unknown. Instead, only their approximate locations, within a certain vicinity, are available. This assumption is both practical and challenging, as it captures real-world conditions where the \gls{RIS} must ensure secure communication without precise location information.
    \item Then, we formulate an optimization problem to design the phase shifts of the \gls{LC}-\gls{RIS} by taking into account the frequency-dependent behavior of each \gls{RIS} element. Since the problem is non-convex, we propose several reformulations that enable the development of an efficient sub-optimal solution based on rank relaxation and semidefinite programming.
    \item Finally, we evaluate the performance of the proposed algorithm against three benchmark schemes, which neglect the frequency dependency of individual \gls{RIS} elements. Simulation results show that our method achieves a higher ensured secrecy rate across different subcarriers in the \gls{OFDM} setup by explicitly accounting for the frequency-dependent characteristics in the \gls{LC}-\gls{RIS} phase shift design.
\end{itemize}

\begin{figure}
    \centering
    \includegraphics[width=0.5\textwidth]{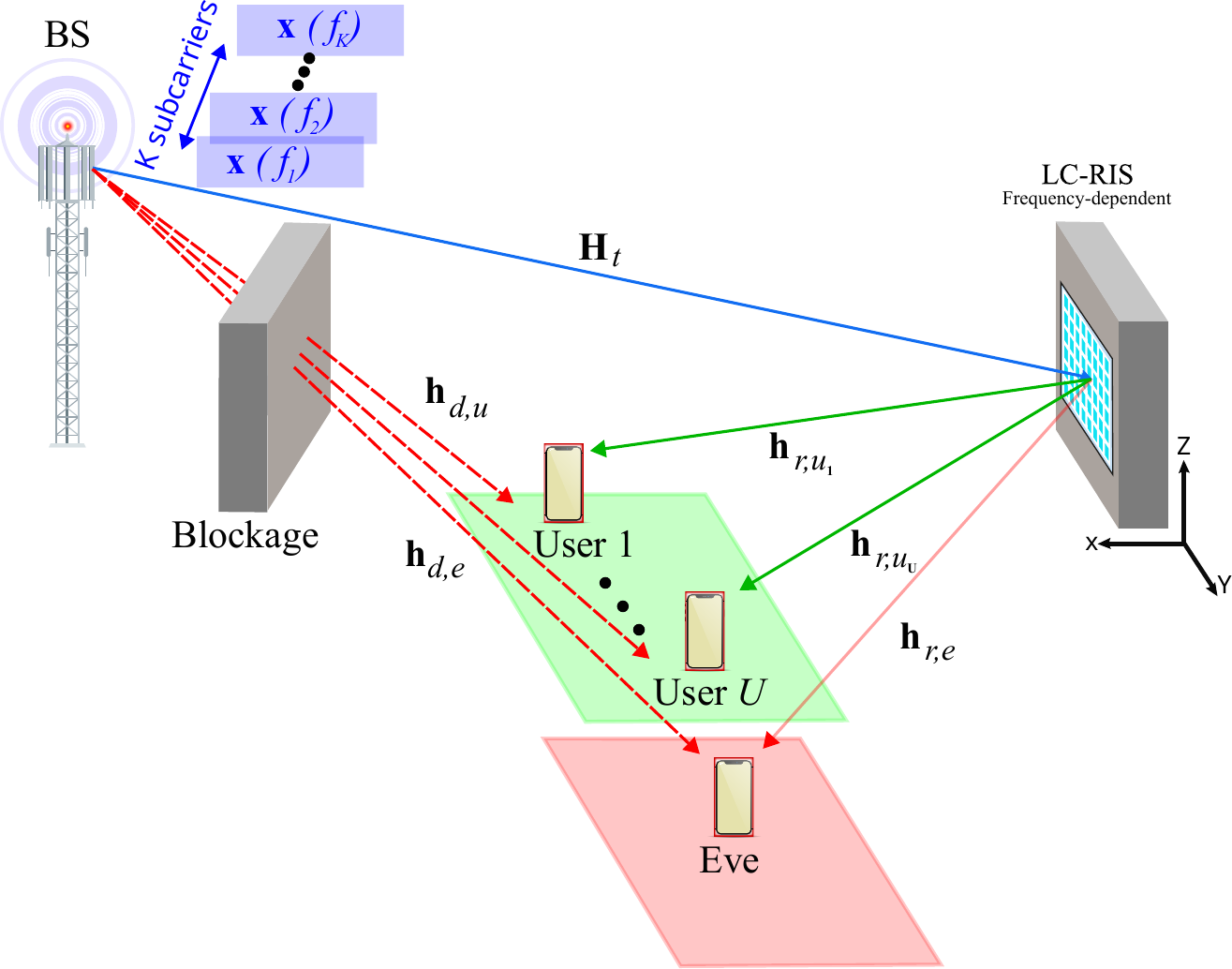}
    \caption{A schematic of a wireless channel model where a \gls{BS} serves a legitimate user via an \gls{RIS} trying to decrease the received signal by an eavesdropper.}
    \label{fig:system model}
\end{figure}

\textit{Notation:} Bold capital and small letters are used to denote matrices and vectors, respectively.  $(\cdot)^\Trans$, $(\cdot)^\Herm$, $\rank(\cdot)$, $\tr(\cdot)$, $\odot$, and $(\cdot)^{\circ}$ denote the transpose, Hermitian, rank, trace of a matrix, Hadamard product, and Hadamard power, respectively. Moreover, $\diag(\bA)$ is a vector that contains the main diagonal entries of matrix $\bA$, $\bone_n$ and $\bzero_n$ denote column vectors of size $n$ whose elements are all ones and zeros, respectively. $\|\bA\|_*=\sum_i \sigma_i$, $\|\bA\|_2=\max_i \sigma_i$, $\|\bA\|_F$, and $\blambda_{\max}(\bA)$ denote the respectively nuclear, spectral, and Frobenius norms of a Hermitian matrix $\bA$, and eigenvector associated with the maximum eigenvalue of matrix $\bA$, where $\sigma_i,\,\,\forall i$, are the singular values of $\bA$. Furthermore, $[\bA]_{m,n}$ and $[\ba]_{n}$ denote the element in the $m$th row and $n$th column of matrix $\bA$ and the $n$th entry of vector $\ba$, respectively. $x^+$ denotes as $\max\{x,0\}$. Moreover, $\Rset$ and $\Cset$ represent the sets of real and complex numbers, respectively, $\jj$ is the imaginary unit, and $\Ex\{\cdot\}$  represents expectation. $\mathrm{rand}(N)$ denotes a $N\times1$ vector where each element is generated independently and uniformly from 0 to 1. $\mathcal{CN}(\bmu,\bSigma)$ denotes a complex Gaussian random vector with mean vector $\bmu$ and covariance matrix $\bSigma$. Finally, $\bigO(\cdot)$ represents the big-O notation and $|\Pset|$ is the cardinality of set $\Pset$.

\section{System, Channel, and LC Models}
\label{sec: System and Channel Models}
In this section, we begin by presenting the system model, which includes $U$ legitimate users and a mobile eavesdropper. Next, we describe the phase-shift model employed for the \gls{LC}-\gls{RIS}. Finally, we introduce the definition of the secrecy rate used throughout this paper.

\subsection{System and Channel Models}
\label{sec: system model}
In this paper, we study a wideband downlink communication scenario involving an \gls{BS} equipped with $N_t$ \gls{Tx} antennas, an \gls{RIS} comprising $N$ \gls{LC}-based unit cells, $U$ single-antenna legitimate users, and a single-antenna eavesdropper. The frequency representations of the receive signal at legitimate \glspl{MU} and \gls{ME} are
\begin{align}
\label{Eq:system model user}
	y_g(f_k) = &\big(\bh_{d,g}^\Herm(f_k) + \bh_{r,g}^\Herm(f_k) \bGamma(f_k) \bH_t(f_k) \big) \bx(f_k)\! +\!n_g,\\
    &\quad\,\, g=\{u_1,\cdots,u_U,e\}, \quad f_k\in\{f_1, \cdots, f_K\}.\nonumber
\end{align}
Let $f_k = f_1, \cdots, f_K$ denote the subcarriers in the \gls{OFDM} setup, where $K$ is the total number of subcarriers. $\bx(f_k)\in\Cset^{N_t}$ is the transmit signal vector for the legitimate users on the $f_k$th subcarrier, $y_u(f_k)\in\Cset$ and $y_e(f_k)\in\Cset$ are the received signal vector at the $u$th legitimate user and eavesdropper, respectively, on subcarrier $f_k$, and  $n_u\in\Cset$ ($n_e\in\Cset$) represents the \gls{AWGN} at the legitimate user (eavesdropper), i.e., $n_u,n_e\sim\sCN(0,\sigma_n^2)$, where $\sigma_n^2$ is the noise power. We adopt hybrid beamforming, where for each subcarrier, the transmit signal vector $\bx(f_k)$ is given by $\bx(f_k) = \bq s(f_k)$ where $\bq\in \Cset^{N_t}$ is beamforming vector on the \gls{BS}, and $s(f_k)\in \Cset^{N_s}$ is the transmitted data symbol at subcarrier $k$ with $\Ex\{|s(f_k)|^2\}=1,\,\forall k$. The beamformer must satisfy the transmit power constraint $\|\bq\|^2 \leq P_t$, where $P_t$ denotes the maximum allowable transmit power. The \gls{RIS} reflection is characterized by a diagonal matrix $\bGamma(f) \in \Cset^{N \times N}$, whose $n$th diagonal entry is given by $[\bGamma]_n(f) = [\bOmega]_n \e^{\jj[\bomega]_n(f)}$, where $[\bomega]_n(f)$ and $[\bOmega]_n$ denote the phase shift and reflection amplitude of the $n$th unit cell, respectively. Based on the design considerations of \glspl{LC}-\glspl{RIS} in the relevant frequency band \cite{neuder2024architecture}, we assume negligible amplitude variation across unit cells, i.e., $[\bOmega]_n \approx 1,\ \forall n$ in all subcarriers.

Moreover,  $\bh_{d,g}(f)\in\Cset^{N_t}, \bH_t(f)\in\Cset^{N\times N_t}$, and $\bh_{r,g}(f)\in\Cset^{N}$ denote the \gls{BS}-\{MU, ME\}, BS-RIS, and RIS-\{MU, ME\} frequency selective channel matrices, respectively, where $g\in\{u_1,\cdots,u_U,e\}$, and $f$ is frequency. We assume that the direct channels between the \gls{BS} and all the legitimate mobile users and the mobile eavesdropper are severely blocked within the coverage area, i.e., $\bh_{d,g}\approx\boldsymbol{0}_{N_t}, \forall g\in\{u_1,\cdots,u_U,e\}$. This motivates the use of an \gls{RIS} to establish a reliable communication link. Due to the use of high-frequency bands, the communication between the \gls{RIS}, \gls{BS}, and the legitimate user is assumed to be dominated by \gls{LOS} components. Accordingly, we adopt a Rician fading model with a high $K$-factor to capture the strong \gls{LOS} path relative to the \gls{nLOS} components. This model is applicable to all channel components $\bH_t$, $\bh_{r,g}$, and $\bh_{d,g}$, for $g\in\{u_1,\cdots,u_U,e\}$, with appropriate adaptations.

\subsection{Frequency Response of an LC Cell}
\label{Frequency response of the LC cell}
In this subsection, we explain how \gls{LC} molecules induce a phase shift on an incoming electromagnetic wave, and how varying the frequency impacts the induced phase shift. \gls{LC}-\glspl{RIS} manipulate signal reflection by exploiting the unique electromagnetic properties of \gls{LC} molecules, which can be reoriented by an externally applied voltage \cite{jimenez2023reconfigurable}. This reorientation changes the permittivity of the \gls{LC} medium, thereby modifying the phase shift introduced by each \gls{RIS} element. \gls{LC} molecules are elongated and rod-shaped, and their electromagnetic response depends on the orientation of the electric field ($\vec{E}_{\rm RF}$) relative to their major and minor axes. When the electric field aligns with the major axis, the permittivity increases, which leads to a larger phase shift. This behavior is illustrated in Fig.~\ref{fig:V_phase}. By tuning the applied voltage, the orientation of the molecules can be controlled. This enables the \gls{RIS} to dynamically modify signal reflections and support a programmable wireless environment.
\begin{figure}[tbp]
	\centering
    \includegraphics[width=0.5\textwidth]{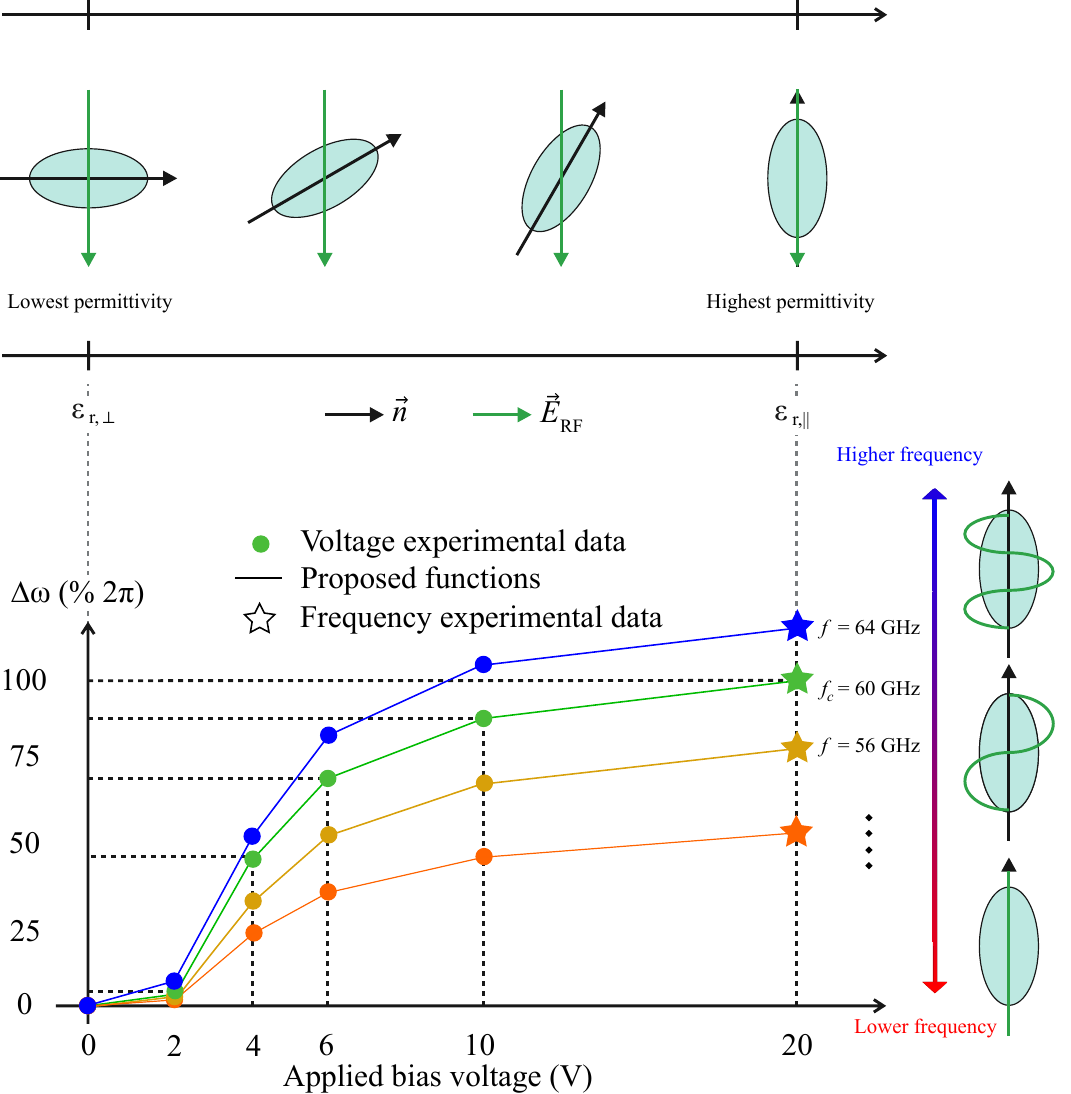}
    \caption{The relationship between phase shift and applied voltage for an \gls{LC} for different frequencies is depicted. The experimental data is given by \cite[Fig.~3.e]{neuder2024architecture}, and a piece-wise linear function is modeled similarly to \cite{delbari2024fast}.}
    \label{fig:V_phase}
\end{figure}
The maximum range of phase shift that an \gls{LC} unit can achieve depends on several factors, such as the length of the \gls{LC} phase shifter, the difference in permittivity between different molecular orientations, and the operating frequency, i.e.,
\begin{equation}
    \Delta\omega_{\max}(f)=2\pi l\Delta n_{\max}\frac{1}{c}\left(f_c+\beta\left(f-f_c\right)\right),
    \label{eq:omega epsilon}
\end{equation}
where $l$, $f$, $f_c$, and $c$ denote the phase shifter length, the operating frequency, the center frequency, and the speed of light in vacuum, respectively. The factor $\beta$ controls how the maximum tuning range of the phase shifter varies with frequency. The term $\Delta n_{\max}$ represents the maximum possible birefringence of the \gls{LC} and is defined as
\begin{equation}
    \Delta n_{\max}=\sqrt{\varepsilon_{r,\parallel}}-\sqrt{\varepsilon_{r,\perp}},
    \label{eq:delta n}
\end{equation}
where $\sqrt{\varepsilon_{r,\parallel}}$ and $\sqrt{\varepsilon_{r,\perp}}$ are the maximum and the minimum relative permittivity, respectively. These values correspond to when an electrical field is aligned or perpendicular to the vector $\vec n$ in Fig. \ref{fig:V_phase}. When the input voltage varies, the value of $0\leq\Delta n\leq\Delta n_{\max}$ changes accordingly. This allows the phase shifters to be controlled at a given frequency. Suppose that the maximum differential phase shift at the center frequency ($f_c$) is $2\pi$. Assuming the minimum (reference) phase shift to be zero, i.e., $\omega_{\min}=0$, we can write the following equation by using \eqref{eq:omega epsilon}:
\begin{equation}
\omega(f)\!=\!2\pi\frac{2\pi l\Delta n\frac{1}{c}\left(f_c+\beta\left(f-f_c\right)\right)}{2\pi l\Delta n_{\max}\frac{f_c}{c}}\!=\!h(v)\!\!\left(\!1\!+\!\beta(\frac{f}{f_c}-1)\!\!\right),
\label{eq:omega reference f}
\end{equation}
where $h(v)\defeq2\pi\frac{\Delta n}{\Delta n_{\max}}$ and $0 \leq h(v) \leq 2\pi,\quad\forall v,$ denotes the externally applied voltage-to-phase relationship that controls the phase shift at a given frequency. While the function $h(v)$ lacks a straightforward analytical expression, it is characterized experimentally and shown in Fig.~\ref{fig:V_phase} (green curve). The relationship between $\omega$ and $f$ is also supported by experimental findings in \cite[Fig.~3e]{neuder2024architecture}. Equation \eqref{eq:omega reference f} indicates that as the frequency increases (or decreases), the maximum achievable phase shift also increases (or decreases) proportionally, as shown in Fig.~\ref{fig:V_phase} and supported by experimental result \cite{neuder2024architecture}. Rather than working directly with voltages, we define a reference phase shift as $\omega_c\defeq\omega(f_c)$ and in Section~\ref{sec: OFDM LC-RIS Phase-shift Design}, we optimize $[\bomega_c]_n,\,\forall n$. However, for the remaining subcarriers, the phase shifts are related by the constraint 
\begin{equation}
    \label{eq: omega to omega_c}
    [\bomega(f_k)]_n=[\bomega_c]_n\!\!\left(\!1\!+\!\beta(\frac{f_k}{f_c}-1)\!\!\right)=[\bomega_c]_n\beta_k,\forall n,\forall k,
\end{equation}
where we defined $\beta_k\defeq\left(\!1\!+\!\beta(\frac{f_k}{f_c}-1)\!\!\right)$ and \eqref{eq: omega to omega_c} ensures that the phase shifts scale consistently across all subcarriers.

\subsection{Secrecy Rate}
\label{sec: Secrecy rate}
The secrecy rate is a crucial metric used to quantify the performance of a physical-layer secure communication system. It is defined as the difference between the achievable rates of the legitimate user and the eavesdropper \cite{delbari2024temperature}: 
\begin{IEEEeqnarray}{ll}
\label{eq: SNR}
    \RS(f_k)\!=[\log(1+\SNR_u(f_k))-\log(1+\SNR_e(f_k))]^+\!,\\
    \SNR_u(f_k)=\frac{|(\bh_u^\eff(f_k))^\Herm\bq|^2}{\sigma^2_n},\\
    \SNR_e(f_k)=\frac{|(\bh_e^\eff(f_k))^\Herm\bq|^2}{\sigma^2_n},
\end{IEEEeqnarray}
where $(\bh_g^\eff)^\Herm=\bh_{d,g}^\Herm + \bh_{r,g}^\Herm \bGamma \bH_t,\,g=\{u_1,\cdots,u_U,e\}$ is the end-to-end channel from the \gls{BS} to the legitimate users and eavesdropper.

In this paper, our goal is to maximize the secrecy rate for all the frequency subcarriers by optimizing the \gls{RIS} phase shifts and \gls{BS} beamformer, thus the system can create a favorable communication environment for the legitimate users while suppressing the signal strength in the area of the eavesdropper\footnote{In this study, we do not consider the subcarrier assignment process in \gls{OFDM}. Instead, we focus on optimizing the phase shifts to ensure good link quality across the entire frequency band (i.e., solving the illumination problem). The orthogonal subcarrier assignment is assumed to be handled separately and is beyond the scope of this work.}. We design the \gls{RIS} phase-shift based on the approximate locations of both the coverage area and the eavesdropper, with the following benefits:
\begin{itemize}
    \item Unlike most literature that assumes the eavesdropper's \gls{CSI} is fully known, we assume only that the eavesdropper’s location lies within a specified area, $\bp_e\in\Pset_e$. The size of this area, $|\Pset_e|$, accounts for potential errors in estimating this location.
    \item Similarly, for the legitimate users, we assume only that their location is within a specific area, $\bp_u\in\Pset_u$. The area size, $|\Pset_e|$, can be adjusted not only to account for location estimation error but also to reduce reconfiguration overhead by decreasing the frequency of required \gls{RIS} adjustments. Hence, we optimize the phase shifts of the \gls{RIS} to serve the legitimate users in all potential locations for all subcarriers at the same time.
  
\end{itemize}
The \gls{RIS} phase shifts will be designed based on the {LOS} paths to optimize the communication, as it offers the most significant contribution to signal strength at \gls{mmWave} bands, where \gls{LC}-\gls{RIS}s are expected to be adopted.

\section{OFDM LC-RIS Phase-shift Design}
\label{sec: OFDM LC-RIS Phase-shift Design}
In this section, we start by formulating an optimization problem for phase-shift design aimed at maximizing the secure rate for all the subcarriers in \gls{OFDM} setup. Consider the setup depicted in Fig. \ref{fig:system model}, where the \gls{RIS} aims to maximize the secure rate defined in \eqref{eq: SNR} for the \glspl{MU} under the worst-case scenario. Specifically, the objective is to maximize the secure rate regardless of the exact locations of all the legitimate users and the eavesdropper, as long as they remain within their respective designated areas. This requirement must be satisfied for all subcarriers $f_k$, where $k = 1, \cdots, K$. We formulate the following problem formulation based on the only \gls{LOS} links:
\begin{subequations}
\label{eq:optimization 1}
\begin{align}
    \text {P1:}\quad&~\underset{\alpha,\bomega_c,\bq}{\max}~\alpha
    \\&~\text {s.t.} ~~\RS(f_k)\geq \alpha,\, \forall \bp_u\in\Pset_u,\,\forall\bp_e\in\Pset_e,\,\forall k,
    \\&\quad\hphantom {\text {s.t.} } 0\leq [\bomega_c]_n < 2\pi,
    \\&\quad\hphantom {\text {s.t.} }  [\bomega(f_k)]_n = [\bomega_c]_n\beta_k, \forall n,\,\forall k,
    \\&\quad\hphantom {\text {s.t.} } \|\bq\|_2^2 \leq P_t.
\end{align}
\end{subequations}
In problem P1, constraint (\ref{eq:optimization 1}b) represents the secure rate condition defined in \eqref{eq: SNR}, (\ref{eq:optimization 1}c) is the realizable phase-shift range for the center frequency, (\ref{eq:optimization 1}d) enforces the phase shifter of each \gls{RIS} element across subcarriers as defined in \eqref{eq: omega to omega_c}, and (\ref{eq:optimization 1}e) limits the transmit power. The parameter $\alpha$ characterizes the worst-case secure rate over all possible user and eavesdropper locations as well as all subcarriers, and the objective is to maximize this value. The optimization problem P1 is non-convex, primarily due to the non-convex nature of constraint (\ref{eq:optimization 1}b). Additionally, the coupling between the variables $\bq$ and $\bomega(f_k),\,\forall k,$ in this constraint further complicates the derivation of a global solution. To address this, we first explain the impact of the number of \gls{Tx} antennas on beam squinting, motivated by the observations in Figs.~\ref{fig: f_N_fixed} and \ref{fig: N_f_fixed} in Section~\ref{Sec: Impact of Number of Antenna in Wideband Scenario}. Next, we decompose the original problem into two sub-problems and iteratively maximize the objective using the \gls{AO} method, as discussed in Section~\ref{sec: Beamformer Design} and Section~\ref{sec: RIS Phase Shifter Design}. Finally, we analyze the overall algorithmic complexity in Section~\ref{sec: Algorithm and Complexity Analysis}.
\subsection{Impact of the Number of \gls{Tx} Antennas on Beam Squinting}
\label{Sec: Impact of Number of Antenna in Wideband Scenario}
Assume a \gls{ULA} is deployed at the \gls{Tx}, located at the origin, with $N$ antennas. A single-antenna \gls{Rx} is positioned at a fixed distance from the \gls{Tx}. Let $\bh_k^\Herm$ denote the channel at the $k$-th subcarrier, and $\bq$ be the transmit beamformer. Suppose we design the beamformer only for the center frequency $f_c$, yielding $\bq_c = \frac{\bh_c}{|\bh_c|} \sqrt{P_t}$. The received \gls{SNR} at subcarrier $k$ is given by $\SNR_k = \frac{|\bh_k^\Herm \bq_c|^2}{\sigma_n^2},$ and we define $\SNR_c$ as the \gls{SNR} at the center frequency, which is the maximum among all subcarriers. In Fig.~\ref{fig: f_N_fixed}, we show the \gls{SNR} of the $k$th subcarrier at frequency $f$ normalized by the \gls{SNR} of the center frequency subcarrier. In Fig.~\ref{fig: N_f_fixed}, we illustrate the minimum normalized \gls{SNR} across all subcarriers. These simulations show the following behavior:
\begin{itemize}
\item When $N$ is fixed and the bandwidth increases, the ratio $\frac{\SNR_k}{\SNR_c}$ decreases, and the rate of this decrease (i.e., the slope) depends on the value of $N$, see Fig.~\ref{fig: f_N_fixed}.
\item When the bandwidth is fixed and the number of \gls{Tx} antennas $N$ increases, the minimum ratio across subcarriers, $\underset{k}{\min} \frac{\SNR_k}{\SNR_c}$, also decreases. However, this degradation remains negligible up to a certain value of $N$, beyond which the distortion becomes significant, see Fig.~\ref{fig: N_f_fixed}.
\end{itemize}

These trends justify approximating the wideband beamformer at the \gls{BS} based on the center frequency response,  when $N$ is relatively small. However, this approximation is not valid for large \glspl{RIS} containing hundreds of elements\footnote{For a more rigorous proof of this claim, see \cite{Parvini}.}.

\begin{figure}[t]
\centering
\begin{subfigure}{0.5\textwidth}
\includegraphics[width=\textwidth,height=0.55\textwidth]{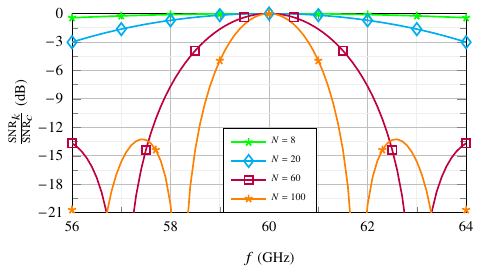}
\caption{The normalized \gls{Rx} power (dB) versus frequency for given $N$ values 8, 20, 60, and 100.}
\label{fig: f_N_fixed}
\end{subfigure}
\begin{subfigure}{0.5\textwidth}
\includegraphics[width=\textwidth,height=0.55\textwidth]{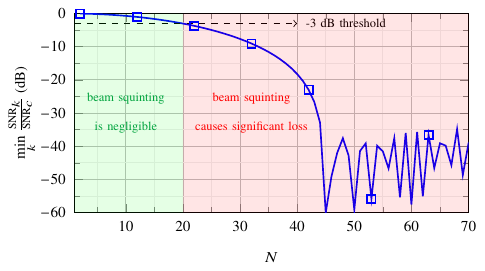}
    \caption{The normalized minimum \gls{Rx} power~(dB) accross all subcarriers versus number of \gls{Tx} antenna for given bandwidth 8 GHz, where $f_c=60$~GHz.}
    \label{fig: N_f_fixed}
\end{subfigure}
\caption{Impact of the bandwidth and number of \gls{Tx} antenna on the \gls{Rx} power.}
\label{fig: N and f different}
\end{figure}

\subsection{Beamformer Design}
\label{sec: Beamformer Design}
In this step, we fix $\bomega_c$ and treat $\bq$ as the sole optimization variable. In this paper, we assume that the number of the \gls{Tx} antenna is in the acceptable regime of Fig.~\ref{fig: N_f_fixed}, so we should only optimize $\bq$ for the center frequency. Furthermore, since the \gls{LOS} path is typically dominant at higher frequencies, we design the beamformer based solely on the \gls{LOS} component. This assumption is generally valid, particularly because both the \gls{BS} and \gls{RIS} are deployed at elevated positions above ground level. Under this assumption, the channel matrix $\bH_t(f_c)$ can be expressed as:
\begin{equation}
\label{eq: H_t}
\bH_t(f_c) = c_0 \ba_\text{RIS}(\bp_\BS,f_c), \ba^\Herm_\BS(\bp_\RIS,f_c),
\end{equation}
where $\ba_\RIS(\cdot)$ and $\ba^\Herm_\BS(\cdot)$ denote the steering vectors at the \gls{RIS} and \gls{BS} in center frequency, respectively, and both satisfy $|\ba_\RIS(\cdot)| = |\ba^\Herm_\BS(\cdot)| = 1$. Given the large \gls{LC}-\gls{RIS}, we adopt the \gls{NF} model for $\ba_\RIS(\cdot)$, as described in \cite{delbari2024nearfield}. Here, $\bp_\BS$ and $\bp_\RIS$ represent the center positions of the \gls{BS} and \gls{RIS}, respectively, while $c_0$ denotes the path loss coefficient of the \gls{LOS} link. With $\bomega$ fixed, the original problem P1 reduces to the following sub-problem:
\begin{subequations}
\label{eq:optimization 2}
\begin{align}
    \text {P2:}\quad&~\underset{\alpha,\bq}{\max}~\alpha
    \\&~\text {s.t.} ~~\RS(f_k)\geq \alpha,\, \forall \bp_u\in\Pset_u,\,\forall\bp_e\in\Pset_e,\,\forall k,
    \\&\quad\hphantom {\text {s.t.} } \|\bq\|_2^2\leq P_t.
\end{align}
\end{subequations}

The optimal beamformer for this sub-problem is characterized in the following lemma.
\begin{lem}
\label{lemma beamforming}
    Assuming a dominant \gls{LOS} link, fixed \gls{RIS} phase shifts, and blocked direct channels for both the legitimate user and the eavesdropper, the optimal beamformer for P2 is given by $\bq=\sqrt{P_t}\ba_\BS(\bp_\RIS,f_c)$.
\end{lem}
\begin{proof}
    The proof is provided in \cite[Lemma~1]{delbari2024temperature} and omitted due to the space constraint.
\end{proof}

\subsection{RIS Phase Shifter Design}
\label{sec: RIS Phase Shifter Design}
 Assuming a fixed beamformer, we aim to maximize the secrecy rate by optimizing the \gls{RIS} phase shifts. Since the number of \gls{RIS} elements is not in the acceptable regime in Fig.~\ref{fig: N_f_fixed}, it is essential to carefully design the \gls{LC}-\gls{RIS} phase shifts across all subcarriers simultaneously. The resulting subproblem for configuring the phase shifts can be formulated as:
\begin{subequations}
\label{eq:optimization 3}
\begin{align}
    \text {P3:}\quad&~\underset{\alpha,\bomega_c}{\max}~\alpha
    \\&~\text {s.t.} ~~\RS(f_k)\geq \alpha,\, \forall \bp_u\in\Pset_u,\,\forall\bp_e\in\Pset_e,\,\forall k,
    \\&\quad\hphantom {\text {s.t.} } 0\leq [\bomega_c]_n < 2\pi,
    \\&\quad\hphantom {\text {s.t.} }  [\bomega(f_k)]_n = [\bomega_c]_n\beta_k, \forall n,\,\forall k.
\end{align}
\end{subequations}
Problem P3 is inherently non-convex due to the non-convexity of constraint (\ref{eq:optimization 3}b) with respect to $[\bomega(f_k)]_n$ for all $n$ and $k$. Without loss of generality, we omit the $[\cdot]^+$ operation from $\RS$ in \eqref{eq: SNR}, and proceed to maximize the secrecy rate using P3. This simplification does not affect the solution: if the resulting $\RS > 0$, the omission is fine; otherwise, if $\RS < 0$, the secrecy rate becomes zero regardless. To reformulate the constraint, we introduce a new variable $\gamma$ such that $\log(\gamma) = \alpha$. Using this substitution, constraint (\ref{eq:optimization 3}b) can be rewritten as:
\begin{equation}
    \frac{1+\SNR_u(f_k)}{1+\SNR_e(f_k)}\geq\gamma\Rightarrow (\SNR_u(f_k)-\gamma\SNR_e(f_k))\geq\gamma-1.
\end{equation}
 To tackle the non-convexity of the problem, we transform P3 into a \gls{SDP} problem. First let us decompose each $\SNR_u(f_k)$ and $\SNR_e(f_k)$ in terms of a vector including exponential of \gls{RIS} phase shifts $\bs_k\defeq[\e^{\jj[\bomega(f_k)]_1}, \cdots, \e^{\jj[\bomega(f_k)]_N}]^\Trans,\forall k$. With this assumption, we have
 \begin{subequations}
    \label{eq: SNR in term of s}
    \begin{align}
        \SNR_u(f_k)=&\bs_k^\Herm\bA^u_k\bs_k,\\
        \SNR_e(f_k)=&\bs_k^\Herm\bA^e_k\bs_k,
    \end{align}
\end{subequations}
where $\bA^g_k=\frac{\diag(\bh_{r,g}(f_k)^\Herm)\bH_t(f_k)\bq\bq^\Herm\bH_t(f_k)^\Herm\diag(\bh_{r,g}(f_k))}{\sigma_n^2},$ $g=\{u_1,\cdots,u_U,e\}$. Then, we define $\bS_k\defeq\bs_k\bs_k^\Herm,\,\forall k$, and $\bA_k(\gamma)\defeq\bA_k^u(\bp_u)-\gamma\bA_k^e(\bp_e)$ where $\bA_k(\gamma)$ is a function of $\bp_u,\,\bp_e,\,f_k$ and $\gamma$. By considering $\bS_c$ as the matrix associated to the center frequency $f_c$, (\ref{eq:optimization 3}d) changes to $\bS_k=\bS_c^{\circ\beta_k}, \forall k$. After applying these reformulations in P3 and because the logarithm is a monotone-increasing function, the problem P3 can change to P4 in the following:
 \begin{subequations}
\label{eq:optimization 4}
\begin{align}
    \text{P4:}&~\underset{\gamma,\bS_c}{\max}~\gamma
    \\&~\text {s.t.} ~~\text{C1: }\tr(\bA_k(\gamma)\bS_c^{\circ\beta_k})\geq\gamma-1, \forall (\bp_u,\bp_e)\in\Pset_u\times\Pset_e,\,\forall k,
    \\&\quad\hphantom {\text {s.t.} }\text{C2: } \bS_k\succeq 0,\,\forall k,
    \\&\quad\hphantom {\text {s.t.} }\text{C3: } \rank(\bS_k)=1,\,\forall k, \\&\quad\hphantom {\text {s.t.} }\text{C4: }\diag(\bS_k)=\bone_N,\,\forall k.
\end{align}
\end{subequations}
This problem is still non-convex due to the non-convexity in
C1 and C3 in $\bS_k,\forall k$, and being coupled $\gamma$ and $\bS_k,\forall k$ in C1. However, before the resolving these non-convexities, first, we introduce Lemma~\ref{lemma S for all k} to simplify P4 and then solve it.

\begin{lem}
\label{lemma S for all k}
   According to equation $\bS_k=\bS_c^{\circ\beta_k}, \forall k$, if constraints C2, C3, and C4 hold for the center frequency, then they hold for all frequencies, i.e., $\forall k$.
\end{lem}

\begin{proof}
We begin with the rank-one constraint in C3. If $\bS_c$ is rank one, it can be written as $\bS_c = \bs_c \bs_c^\Herm$. Then, we have:
\[
\bS_k = \bS_c^{\circ \beta_k} = (\bs_c \bs_c^\Herm)^{\circ \beta_k} = \bs_c^{\circ \beta_k} (\bs_c^{\circ \beta_k})^\Herm,
\]
which is also rank one.

Next, we verify that $\bS_k$ is positive semidefinite as in C2, assuming that $\bS_c \succeq 0$. Since $\bS_c \succeq 0$, we know that $\bx^\Herm \bS_c \bx \geq 0,\,\forall \bx$. Then:
\[
\bx^\Herm \bS_k \bx = \tr(\bx^\Herm \bS_k \bx) = \tr(\bs_c^{\circ \beta_k} (\bs_c^{\circ \beta_k})^\Herm \bx \bx^\Herm) \geq 0,\quad \forall \bx,
\]
which confirms that $\bS_k$ is also positive semidefinite.

Finally, consider the diagonal constraint in C4. If $\diag(\bS_c) = \bI_N$, then $[\bS_c]_{i,i} = 1,\,\forall i$, and thus $[\bS_c]_{i,i}^{\beta_k} = 1^{\beta_k}=1,\,\forall k$. Hence, $\diag(\bS_k) = \bI_N,\, \forall k$. This completes the proof.
\end{proof}

Using Lemma~\ref{lemma S for all k}, we can rewrite problem P4 as follows:
 \begin{equation}
\label{eq:optimization 5}
\begin{aligned}
    \text {P5:}&~\underset{\gamma,\bS_c}{\max}~\gamma
    \\&~\text {s.t.} ~~\text{C1},\, \widehat{\text{C2}}: \bS_c\succeq 0,\,
    \\&\quad\hphantom {\text {s.t.} } \widehat{\text{C3}}: \rank(\bS_c)=1,\,  \widehat{\text{C4}}: \diag(\bS_c)=\bone_N.
\end{aligned}
\end{equation}
This problem is still non-convex due to the non-convexity in C1 and $\widehat{\text{C3}}$ in $\bS_c$, and being coupled $\gamma$ and $\bS_c$ in C1. We resolve the non-convexity of each one in the following.
\subsubsection{Rank one constraint in  $\widehat{\text{C3}}$} To tackle this issue, we adopt the penalty method exploited in \cite{delbari2024far}. The basic idea is to replace the rank constraint with inequality $\|\bS\|_*-\|\bS\|_2 \leq 0$, which holds only if $\bS$ has rank smaller or equal to one. While the new constraint is still non-convex, one can apply the first-order Taylor approximation to make it convex. Let $\bS^{(i)}$ denotes the value of matrix $\bS$ in the $i$th iteration. According to the first-order Taylor approximation:
\begin{equation}
\label{eq: taylor approximation}
    \|\bS\|_2\geq\|\bS^{(i)}\|_2+\tr\big(\blambda_{\max}(\bS^{(i)})\blambda_{\max}^\Herm(\bS^{(i)})(\bS-\bS^{(i)})\big).
\end{equation}
By adopting the penalty method and applying \eqref{eq: taylor approximation} into the cost function of P5, we have
\begin{subequations}
\label{eq:optimization 6}
\begin{align}
    \text {P6:}\quad&~\underset{\gamma,\bS_c}{\max}~\gamma-\eta^{(i)}\Big(\|\bS_c\|_*-\|\bS_c^{(i)}\|_2-\tr\big(\blambda_{\max}(\bS_c^{(i)}),\nonumber
    \\&\quad\quad\quad\times\blambda_{\max}^\Herm(\bS_c^{(i)})(\bS_c-\bS_c^{(i)})\big)\Big)
    \\&~\text {s.t.} ~~\text{C1}, \widehat{\text{C2}}, \widehat{\text{C4}}.
\end{align}
\end{subequations}
Here, $\eta^{(i)}$ is the penalty factor at iteration $i$, which increases gradually. By selecting a sufficiently large $\eta$, problems P5 and P6 become equivalent. In the following, we will address the non-convexity of C1 in $\bS_c$.

\subsubsection{Hadamard power constraint on $\bS_c$ in C1} To tackle with this issue, we substitute $\bS_c^{\circ\beta_k}$ with its first order Taylor approximation\footnote{Higher-order Taylor approximations are also applicable; however, they increase computational complexity. We choose the first-order approximation to maintain simplicity.} as follows:
\begin{equation}
    \bS_c^{\circ\beta_k}\approx{\left(\bS_c^{(i)}\right)}^{\circ(\beta_k)}+\left(\beta_k{\left(\bS_c^{(i)}\right)}^{\circ(\beta_k-1)}\right)\odot(\bS_c-\bS_c^{(i)}),\,\forall k,
\end{equation}
where $i$ denotes $i$th iteration. By substituting this approximation in C1, we can define $\widehat{\text{C1}}$ as follows:
\begin{align}
&\tr\!\left(\!\!\bA_k(\gamma)\left(\!\!{\left(\bS_c^{(i)}\right)}^{\circ(\beta_k)}\!\!\!\!+\!\!\left(\beta_k{\left(\bS_c^{(i)}\right)}^{\circ(\beta_k-1)}\right)\odot\left(\bS_c-\bS_c^{(i)}\right)\!\!\right)\!\!\right)\nonumber\\
&\geq\gamma-1, \forall (\bp_u,\bp_e)\in\Pset_u\times\Pset_e,\,\forall k.
\end{align}
By substituting $\widehat{\text{C1}}$ instead of C1 in P6, we have:
\begin{subequations}
\label{eq:optimization 7}
\begin{align}
    \text {P7:}\quad&~\underset{\gamma,\bS_c}{\max}~\gamma-\eta^{(i)}\Big(\|\bS_c\|_*-\|\bS_c^{(i)}\|_2-\tr\big(\blambda_{\max}(\bS_c^{(i)}),\nonumber
    \\&\quad\quad\quad\times\blambda_{\max}^\Herm(\bS_c^{(i)})(\bS_c-\bS_c^{(i)})\big)\Big)
    \\&~\text {s.t.} ~~\widehat{\text{C1}}, \widehat{\text{C2}}, \widehat{\text{C4}}.
\end{align}
\end{subequations}
\subsubsection{Coupled $\gamma$ and $\bS_c$ in $\widehat{\text{C1}}$} To address the coupling of $\gamma$ and $\bS_c$ in $\widehat{\text{C1}}$, we employ \gls{AO}, where one variable is fixed while the other is optimized. On one hand, when $\gamma$ is fixed, the optimization problem P7 becomes convex in the
matrix $\bS_c$ because the objective function is concave, and the
constraints define a convex set. Therefore, it can be efficiently
solved using standard convex optimization solvers such as
CVX \cite{cvx}. On the other hand, when $\bS_c$ is fixed, The problem P7 is linear in terms of $\gamma$ and its closed-form solution is given
by:
\begin{equation}
\label{eq: best gamma}
    \gamma=\underset{\forall \bp_u\in\Pset_u,\,\forall\bp_e\in\Pset_e,\,\forall k}{\min}~\frac{\tr(\bA_k^u(\bp_u)\bS_k)+1}{\tr(\bA_k^e(\bp_e)\bS_k)+1}.
\end{equation}

\subsection{Algorithm and Complexity Analysis}
\label{sec: Algorithm and Complexity Analysis}
The proposed algorithm is summarized in Algorithm \ref{alg:cap}. At each iteration, the most computational step is the computation of the nuclear norm in line 4, which has a complexity of $\bigO(N^3)$. The number of different constraints generated by $\widehat{\text{C1}}$ is the bottleneck and proportional to $|\Pset_u||\Pset_e|K$. Thus, the complexity of the Algorithm \ref{alg:cap} in total is $\bigO(I_{\max}|\Pset_u||\Pset_e|N^3K)$.

\begin{algorithm}[t]
\caption{Proposed Algorithm for Problem P7}\label{alg:cap}
\begin{algorithmic}[1]
\STATE \textbf{Initialize:} $\bs_c^{(0)}=\e^{\jj2\pi\times\mathrm{rand}(N)},\bS_c^{(0)}=\bs_c^{(0)}{\bs_c^{(0)}}^\Herm$, $\gamma^{(0)}$.
\FOR{$j=1, \cdots, J_{\tmax}$}
\STATE Calculate $\bA_k(\gamma^{(j-1)}),\forall (\bp_u,\bp_e)\in\Pset_u\times\Pset_e,\,\forall k$.
    \FOR{$i=1, \cdots, I_{\tmax}$}
    \STATE Solve convex P7 for given $\bS_c^{(i-1)}$, and store the intermediate solution $\bS_c^{(i)}$.
    \STATE Update $\eta^{(i)} =5\eta^{(i-1)}$ then set $i = i + 1$.
    \ENDFOR
    \STATE Set $\bS_k=\left(\bS_c^{(I_{\max})}\right)^{\circ(\beta_k)},\,\forall k,$ and $\bS_c^{(0)}=\bS_c^{(I_{\max})}$.
    \STATE Calculate $\gamma^{(j)}$ from \eqref{eq: best gamma}.
    
    \ENDFOR
\end{algorithmic}
\end{algorithm}

\section{Performance Comparison}
\subsection{Simulation Setup}
We adopt the simulation configuration for coverage extension illustrated in Fig. \ref{fig:system model}, where the \gls{RIS} center is the origin of the Cartesian coordinate system, i.e., $[0,0,0]~\text{m}$. A set of legitimate users is uniformly distributed within a fixed region defined as $\Pset_u\in\{(\x,\y,\z):5~\text{m}\leq\x\leq 7~\text{m}, 0~\text{m}\leq\y\leq 2~\text{m}, \z=-5\}$. We also estimated a fixed area for eavesdropper in $\Pset_e\in\{(\x,\y,\z):5~\text{m}\leq\x\leq 6~\text{m}, -2~\text{m}\leq\y\leq -1~\text{m}, \z=-5\}$. The \gls{BS} is equipped with an $16 \times 16 = 256$-element \gls{UPA}, arranged along the $x$–$z$ plane and located at $[10, 10, 5]~\text{m}$. The \gls{RIS} is modeled as a \gls{ULA} with $N_y = 100$ elements aligned along the $y$-axis. The inter-element spacing for both the \gls{BS} and \gls{RIS} arrays is set to half the wavelength, $\lambda_c=5~$mm. The noise power is calculated as $\sigma_n^2=W_kN_0N_{\rm f}$  with $N_0=-174$~dBm/Hz, $N_{\rm f}=6$~dB, and bandwidth of each subcarrier $W_k=4.2~$MHz. We choose $60$~GHz carrier frequency and the total bandwidth $W=8.64~$GHz in accordance with the IEEE 802.11ay standard \cite{ieee80211ay}. The pathloss for the $k$th subcarrier is modeled as $\rho_k(d_0/d)^\sigma$, where $\rho_k=(\frac{c}{4\pi f_k})^2$ at $d_0=1$~m, and $c$ is the speed of light in vacuum. Moreover, we adopt the pathloss exponent $\sigma = (2,2,2)$ and  Ricean $K$-factor $K=(0,10,10)$ for the \gls{BS}-\gls{MU}, \gls{BS}-\gls{RIS}, and \gls{RIS}-\gls{MU} channels, respectively. 
To evaluate performance, we compare the proposed method against three benchmarks:
\begin{itemize}
    \item \textbf{Benchmark 1:} The \gls{RIS} phase shifts are optimized to distribute reflected power across the entire coverage area for all subcarrier frequencies while suppressing power in the eavesdropper region. This approach does not account for the frequency-dependent phase response of individual \gls{RIS} elements.
    \item \textbf{Benchmark 2:} The \gls{RIS} phase shifts are designed to maximize coverage area power and minimize eavesdropper region power, but only at the center frequency $f_c = 60$~GHz.
    \item \textbf{Benchmark 3:} The \gls{RIS} phase shifts are optimized across all subcarrier frequencies, but only based on the exact locations of users and the eavesdropper, without explicitly shaping the beam across a broader area.
\end{itemize}

 The other parameters used in the simulations are as follows: $P_t=43~$dBm, $\beta=2.4$ \cite{neuder2024architecture}, $\eta^{(0)}=0.01$, $I_{\max}=9$, and $J_{\max}=2$ are assumed.

\subsection{Simulation Results}
Fig.~\ref{fig:SR_f} illustrates the secrecy rate (bits/symbol) \gls{wrt} frequency, as defined in \eqref{eq: SNR}, under the assumption of worst-case received power for legitimate users within the region $\Pset_u$ and maximum received power for the eavesdropper within the region $\Pset_e$. As shown in the figure, the minimum secrecy rate achieved by our proposed method across all frequencies consistently exceeds that of the benchmark methods. 

Benchmark 3 distributes the reflected power across all frequencies; however, since its design depends on the exact user and eavesdropper locations, its worst-case performance is inferior to the others. Benchmark 2 concentrates power within the target area but only at the center frequency, which results in a secrecy rate that degrades as the signal frequency deviates from the center. Finally, Benchmark 1 optimizes the phase shifts based on the signal frequency and the considered areas, but it neglects the frequency-dependent phase response of individual \gls{RIS} elements (similar to the other benchmarks).

\begin{figure}
    \centering
    \includegraphics[width=0.5\textwidth]{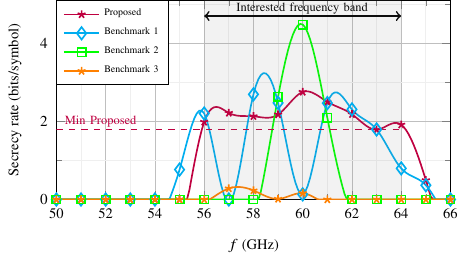}
    \caption{The minimum secrecy rate over all possible locations $\bp_u\in\Pset_u$ and $\bp_e\in\Pset_e$ versus frequency.}
    \label{fig:SR_f}
\end{figure}

To illustrate how the \gls{SNR} is distributed spatially at a specific frequency, we present heat maps of the \gls{SNR} for all four methods in Fig.~\ref{fig: heat map f=58}, corresponding to the carrier frequency $f_k = 57$~GHz. The coverage area and the eavesdropper area are indicated by green and red rectangles, respectively. In Fig.~\ref{fig: benchmark 3}, corresponding to Benchmark 3, the \gls{SNR} within the desired area is non-uniform, as this method focuses only on frequency and not on spatial regions. This non-uniformity leads to degraded performance. Fig.~\ref{fig: benchmark 2} shows the result of Benchmark 2, where noticeable beam splitting occurs due to its optimization being limited to the center frequency. As a result, the \gls{SNR} in the eavesdropper area increases. The outcome for Benchmark 1 is illustrated in Fig.~\ref{fig: benchmark 1}. Since it ignores the frequency-dependent phase response of each \gls{RIS} element, the \gls{SNR} decreases in parts of the desired area and, conversely, increases in the eavesdropper area, thereby reducing its overall performance. In contrast, the proposed method (Fig.~\ref{fig: proposed}) effectively distributes \gls{SNR} uniformly across the desired area while simultaneously minimizing it in the eavesdropper area, demonstrating its superior spatial and frequency-aware design.

\begin{figure}[t]
\centering
\begin{subfigure}{0.24\textwidth}
    \caption{Proposed}
\includegraphics[width=\textwidth,height=0.7\textwidth]{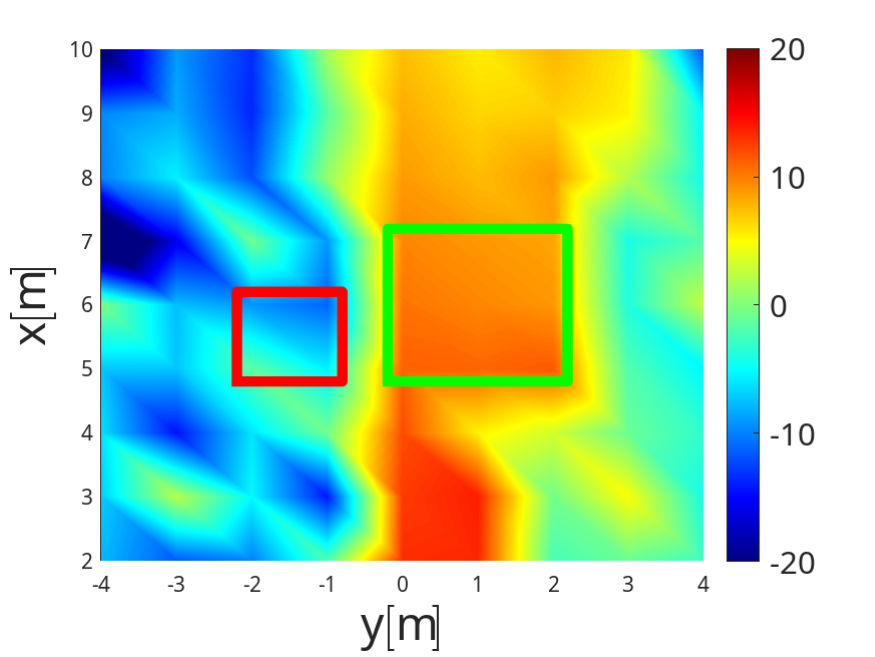}
    \label{fig: proposed}
\end{subfigure}
\begin{subfigure}{0.24\textwidth}
    \caption{Benchmark 1}
    \includegraphics[width=\textwidth,height=0.7\textwidth]{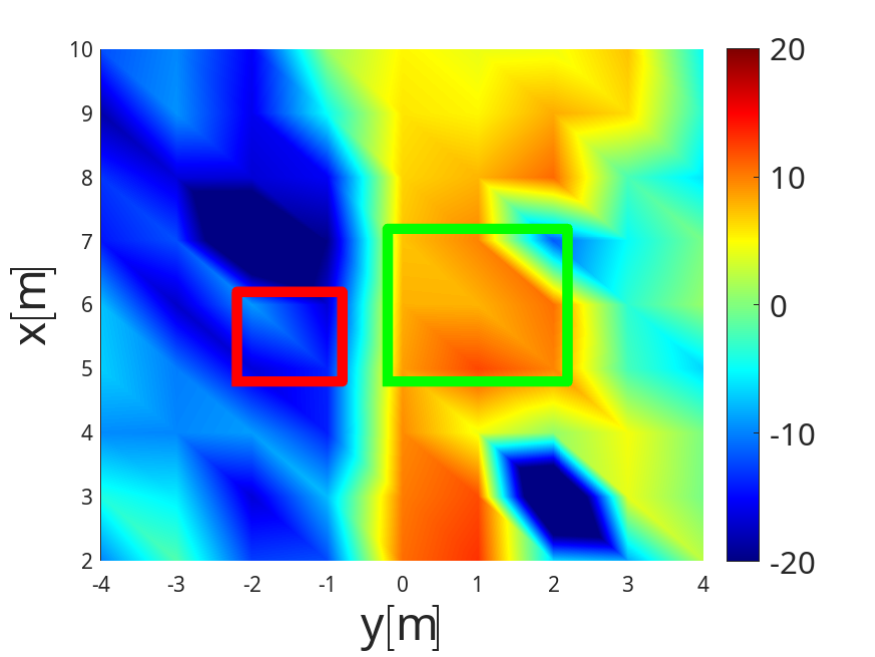}
    \label{fig: benchmark 1}
\end{subfigure}
\begin{subfigure}{0.24\textwidth}
    \caption{Benchmark 2}
   \includegraphics[width=\textwidth,height=0.7\textwidth]{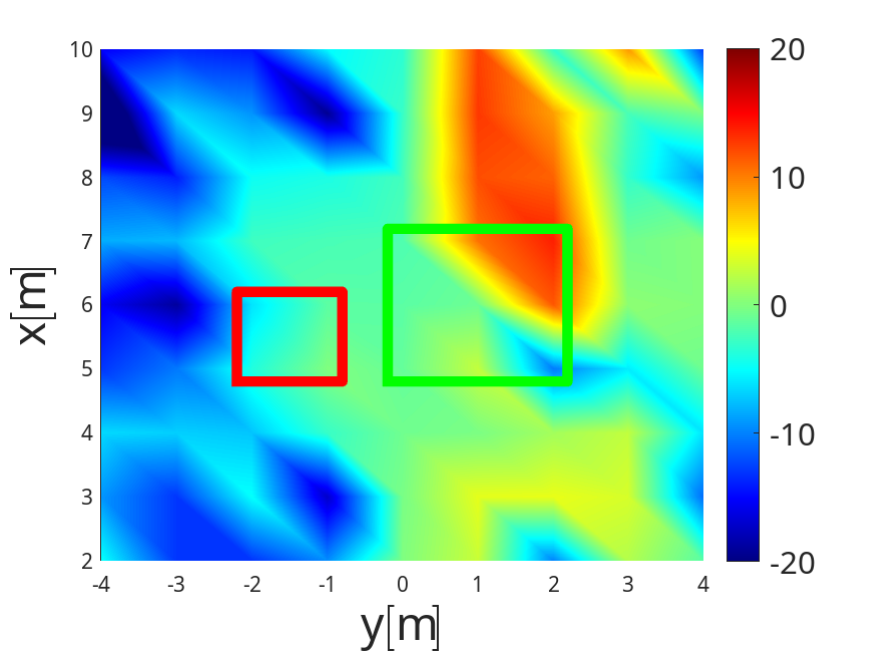}
    \label{fig: benchmark 2}
\end{subfigure}
\begin{subfigure}{0.24\textwidth}
    \caption{Benchmark 3}
    \includegraphics[width=\textwidth,height=0.7\textwidth]{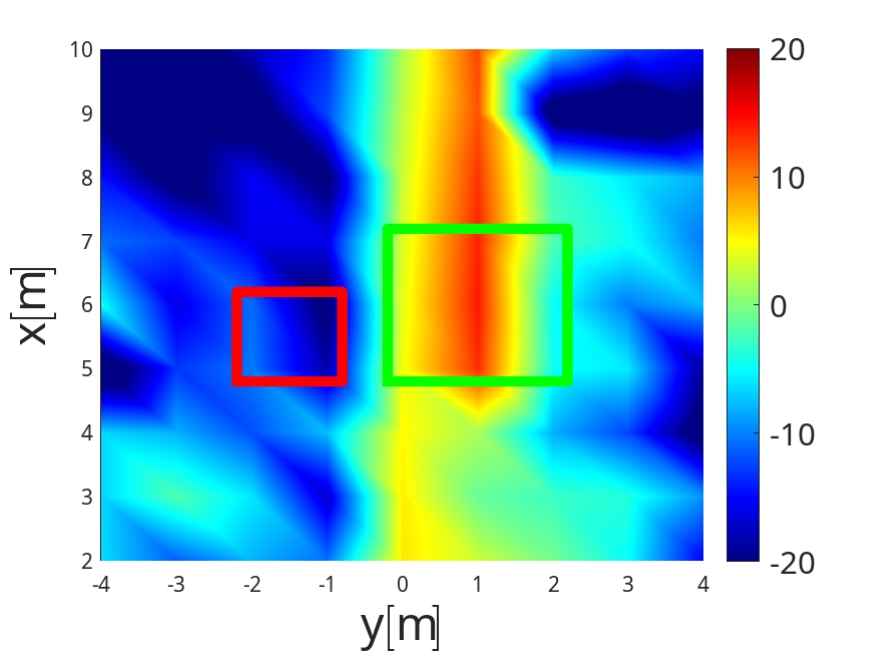}
    \label{fig: benchmark 3}
\end{subfigure}
\caption{\gls{SNR} (dB) using four different methods at $f_k = 57$~GHz. The green area represents the user coverage region, while the red area indicates the possible locations of the eavesdropper.}
\label{fig: heat map f=58}
\end{figure}

\section{Conclusion}
\label{sec: Conclusion}
In this paper, we investigated the impact of frequency variation on the phase shifts of \gls{RIS} elements. To address the resulting frequency dependency, we proposed an algorithm that accounts for this effect in the design of phase shifts, aiming to maximize the secrecy rate within a desired area. Simulation results confirmed the importance of incorporating the frequency-dependent behavior of \gls{RIS} elements, which demonstrates significant performance gains when this factor is properly considered. Specifically, the proposed method achieved a secrecy rate of about 2 bits/symbol over an 8 GHz bandwidth at a center frequency of 60 GHz.

\bibliographystyle{IEEEtran}
\bibliography{References}

\end{document}